\documentclass[bimj,fleqn]{w-art}
\usepackage{times}
\usepackage{w-thm}
\usepackage[authoryear]{natbib}
\theoremstyle{plain}

\theoremstyle{definition}

\usepackage[]{graphicx}
\chardef\bslash=`\\ 

\usepackage{graphicx}
\usepackage{multirow}
\usepackage{booktabs}
\usepackage{endfloat}

\usepackage{calc}
\usepackage{lipsum}

\usepackage{psfrag}
\usepackage{epsfig}
\usepackage{url}
\usepackage{setspace}
\usepackage{amsmath}
\usepackage{amsfonts, amssymb}

\usepackage{tikz} 

\usepackage{algorithm}
\usepackage{algpseudocode}

\usepackage{xr} 
\usepackage{hyperref}
\usepackage{pifont}
\usepackage{xcolor}
\usepackage{mathtools}

\makeatletter
\newcommand*{\addFileDependency}[1]{%
  \typeout{(#1)}%
  \@addtofilelist{#1}%
  \IfFileExists{#1}{}{\typeout{No file #1.}}
}
\makeatother

\newcommand*{\myexternaldocument}[1]{%
  \externaldocument{#1}%
  \addFileDependency{#1.tex}%
  \addFileDependency{#1.aux}%
}

\myexternaldocument{supp}

\newcommand{\cross}{\ding{53}}%
\newcommand{\rstar}{\ding{84}}%
\newcommand{\tri}{\ding{115}}%
\definecolor{amber}{rgb}{1.0, 0.75, 0.0}
\newtheorem{lemma2}{Lemma}
\newtheorem{remark2}{Remark}
\DeclareMathOperator*{\argmax}{arg\,max}

\begin{document}
\DOIsuffix{bimj.200100000}
\Volume{52}
\Issue{61}
\Year{2024}
\pagespan{1}{}

\title[Two-Stage Multiple Testing with auxiliary variable]{Two-Stage Multiple Test Procedures Controlling False Discovery Rate with auxiliary variable and their Application to Set4$\Delta$ Mutant Data}

\author[Hwang {\it{et al.}}]{Seohwa Hwang\inst{1}}
\author[Ramos]{Mark Louie Ramos\inst{2}}
\author[Park]{DoHwan Park\inst{3}}
\author[Park]{Junyong Park\footnote{ Corresponding author: {\sf{e-mail: junyongpark@snu.ac.kr}}}\inst{1}} 
\author[Lim]{Johan Lim\inst{1}} 
\author[Green]{Erin Green\inst{4}}

\address[\inst{1}]{Department of Statistics, Seoul National University, Seoul, Korea}
\address[\inst{2}]{Department of Health Policy and Administration, The Pennsylvania State University, PA, USA}
\address[\inst{3}]{Department of Mathematics and Statistics, University of Maryland Baltimore County, Baltimore, MD, USA}
\address[\inst{4}]{Department of Biological Sciences, University of Maryland Baltimore County, Baltimore, MD, USA}

\Receiveddate{zzz} \Reviseddate{zzz} \Accepteddate{zzz} 

\begin{abstract} 
In this paper, we present novel methodologies that incorporate auxiliary variables for multiple hypotheses testing related to the main point of interest while effectively controlling the false discovery rate. When dealing with multiple tests concerning the primary variable of interest, researchers can use auxiliary variables to set preconditions for the significance of primary variables, thereby enhancing test efficacy. Depending on the auxiliary variable's role, we propose two approaches: one terminates testing of the primary variable if it does not meet predefined conditions, and the other adjusts the evaluation criteria based on the auxiliary variable. Employing the copula method, we elucidate the dependence between the auxiliary and primary variables by deriving their joint distribution from individual marginal distributions.Our numerical studies, compared with existing methods, demonstrate that the proposed methodologies effectively control the FDR and yield greater statistical power than previous approaches solely based on the primary variable. As an illustrative example, we apply our methods to the Set4$\Delta$ mutant dataset. Our findings highlight the distinctions between our methodologies and traditional approaches, emphasising the potential advantages of our methods in introducing the auxiliary variable for selecting more genes.
\end{abstract}

\maketitle                   






\section{Introduction}
In this paper, we consider multiple testing scenarios where additional information beyond p-values is available. Numerous studies have already tackled this subject, predominantly focusing on enhancing rejection rates through the utilisation of covariates. For instance, \cite{Lei2018} demonstrated that regardless of covariate values, $p$-values under the null hypothesis conform to a uniform distribution. In contrast, AdaFDR, proposed by \cite{Zhang2019}, relaxes the requirement of a uniform distribution under the null hypothesis but mandates a mirroring property in the distribution of $p$-values given specific covariates, which means the conditional null distribution of $p$-value given the covariate is symmetric around $1/2$. \cite{Ignatiadis} explore maximising rejections by assigning varying weights to groups based on covariates, while \cite{Boca-Leek} estimate the proportion of alternative hypotheses given covariates to adjust the rejection region. 
On the other hand, \cite{Cao2022} proposed a multiple testing procedure when there is auxiliary information such as the order of hypotheses or alternative distribution while there is no specific covariate. These methodologies aim to optimise rejection rates while ensuring control over the False Discovery Rate (FDR) across all hypothesis tests. While maximising the number of rejections with controlling an error rate is statistically advantageous, there are instances where such strategies may not be desirable when considering scientific knowledge or background.
 We incorporate the auxiliary variable more relevantly through modelling the joint distribution of the $p$-value and the auxiliary variable so that we do not require such strong assumptions 
 such as uniform distribution of $p$-value under the null hypothesis regardless of the covariate or mirroring property of $p$-value under the null hypothesis.

One main aim is to incorporate an auxiliary variable to improve the results of multiple testing compared to the case of using only the primary variable. In addition, scientists may reflect some prior knowledge on the main interest through the auxiliary variable. Some restrictions on the distribution of $p$ value given the value of the covariate may not reflect such flexible situations. The reflection of the auxiliary variable in the test of the primary variable essentially signifies that there is a dependency between the two variables, and identifying this dependency is crucial. We emphasise that only the primary variables are tested and the auxiliary variables provide supplementary information to make better decisions on the primary variables. In this sense, our problem is different from the case of bivariate $p$-values which are declared to be significant when at least one of $p$-values is significant. See \cite{du2014single}.

Rather than imposing the restrictions on the distribution of the $p$ value, we use the joint distribution of the auxiliary and primary variables and then propose two types of two-stage procedures controlling FDR. One is called two-stage FDR(H) based on the idea of hard thresholding such as we filter insignificant genes using the auxiliary variable and then only if this stage is passed, the significance of the primary variable is considered. The other one is two-stage FDR(S) where (S) presents the soft threshold which avoids filtering. Two-stage FDR(S) is testing of the primary variable with adjusted rejection region for the primary variable reflecting the result for the auxiliary variable. As mentioned, our goal is to test only the primary variable while reflecting the information coming from the dependence between the primary and auxiliary variables. The proposed methods are demonstrated through simulation to effectively control the FDR while achieving superior testing power compared to existing methods. Existing methods incorporating the auxiliary variable with the restrictions on the distribution of $p$ value have difficulty in controlling a given level of FDR and obtaining powers due to the constraints such as the uniform distribution of $p$ value regardless of the value of covariate and the mirroring property under the null hypothesis. In the context of Set4 $\Delta$ dataset, the two-stage FDR(S) approach successfully identifies a cluster of genes that not only possess statistical significance but are also known as significant from biological experiments.

This paper is organised as follows. In section 2, we present the description of Set4$\Delta$ mutant data, which serves as the motivation for this study, along with the necessity of the two-step method proposed. Section 3 introduces the copula that explains distribution estimation and dependence, which are the basic items required for the proposed methods, and explains the joint distribution estimation procedure. Section 4 describes two methodologies based on the two-step procedures designed for controlling FDR. Section 5 provides simulation studies to compare the existing methods with the proposed two-stage procedures and section 6 presents an analysis of Set4$\Delta$ mutant data. Finally, Section 7 contains concluding remarks.

\section{Motivating Data set: Set4$\Delta$ mutant data}
\label{sec:data}
The motivating study seeks to identify genes within the genome of the budding yeast {\it Saccharomyces cerevisiae} that are differentially expressed when a particular gene expression machinery referred to as Set4 is mutated. The data set is referred to as the Set4$\Delta$ mutant data set. In order to identify the genes that depend on Set4 during stress, the gene expression profile of the whole genome in cells where Set4 is silenced (knock-out condition or KO) and cells where it is not (wildtype condition or WT) are determined. In addition, cells are treated with hydrogen peroxide to induce oxidative stress. The experiment is replicated thrice. In each replication, the amount of mutation, represented by a weighted count value, was collected for each of about 8,000 RNA species under KO and WT conditions. For each gene, these weighted counts are averaged to produce a single value under each condition.
More formally, we have the data for $M$ genes and three replicates such as
\begin{eqnarray}
WT_{ij}, \quad KO_{ij}
\end{eqnarray}
under WT and KO conditions, respectively for
$1\leq i \leq M$ and $1\leq j \leq 3$.
Let $\mu_{WT_i}$ and $\mu_{KO_i}$
be the mean of $W_{ij}$ and $K_{ij}$, respectively, then we define the vector of
logfold changes of $M$ genes as
\begin{eqnarray}
\beta' = (\beta_1,\ldots, \beta_M)'
= \left(\log \frac{\mu_{WT_1}}{\mu_{KO_1}},\ldots,
\log \frac{\mu_{WT_M}}{\mu_{KO_M}} \right)'.
\end{eqnarray}
From these data matrices,
the vector of statistics that is to be used for multiple testing is
\begin{equation}
\hat \beta'=
(\hat \beta_1,\ldots, \hat \beta_M)' =
\left(\log{\frac{\bar{WT_{1.}}}{\bar{KO_{1.}}}},
\cdots , \log{\frac{\bar{WT_{M.}}}{\bar{KO_{M.}}}}\right)'
\label{eqn
}
\end{equation}
where $\bar{WT_{i.}} = \frac{1}{3}\sum_{j=1}^3 W_{ij}$ and $\bar{KO_{i.}} = \frac{1}{3}\sum_{j=1}^3 KO_{ij}$
are the sample means from 3 replicates for the $i$th gene.
The logfold changes statistic is then computed as
\begin{eqnarray}
\widehat{\beta_i}=\log_2\frac{\bar{KO}i}{\bar{WT}}{i}
\end{eqnarray}
where $i=1,2, \cdots M$ and $M$ is the total number of genes. $\bar{WT}$ is a mean weighted count measure of mutations under wild type condition and $\bar{KO}$ is a mean weighted count measure of mutations under knock-out condition for the $i$th gene.
This statistic $\hat \beta_i$ represents the doubling effect'' of one condition relative to the other condition. A gene is considered interesting'' if the weighted count measure is very different between the KO and WT settings, and so correspondingly, if
$\mid \hat \beta_i \mid\gg0$,
we have supporting evidence for significance of the gene.
\cite{ramos2021}
selected genes based on
$\hat \beta_i$s
considering multiple testing procedures controlling FDR.

While logfold change is our primary focus,
\cite{ramos2021} presented the results that
the genes selected based solely on logfold changes are not
biologically significant
in most cases. Even if logfold change is considered significant for various reasons, such as extremely small values or high variance in the data across conditions, it can be influenced by various errors occurring during the experimental process. Therefore, when considering variability in experimental replicates, it is essential to first assess the significance of changes in logfold changes values. In our analysis, the standard deviation of the logfold changes will serve as an auxiliary variable. {The distributions of logfold changes and their standard deviations are shown in Figure \ref{fig:hist_margin}.
\begin{figure}[ht]
    \centering
    \includegraphics[width = \textwidth]{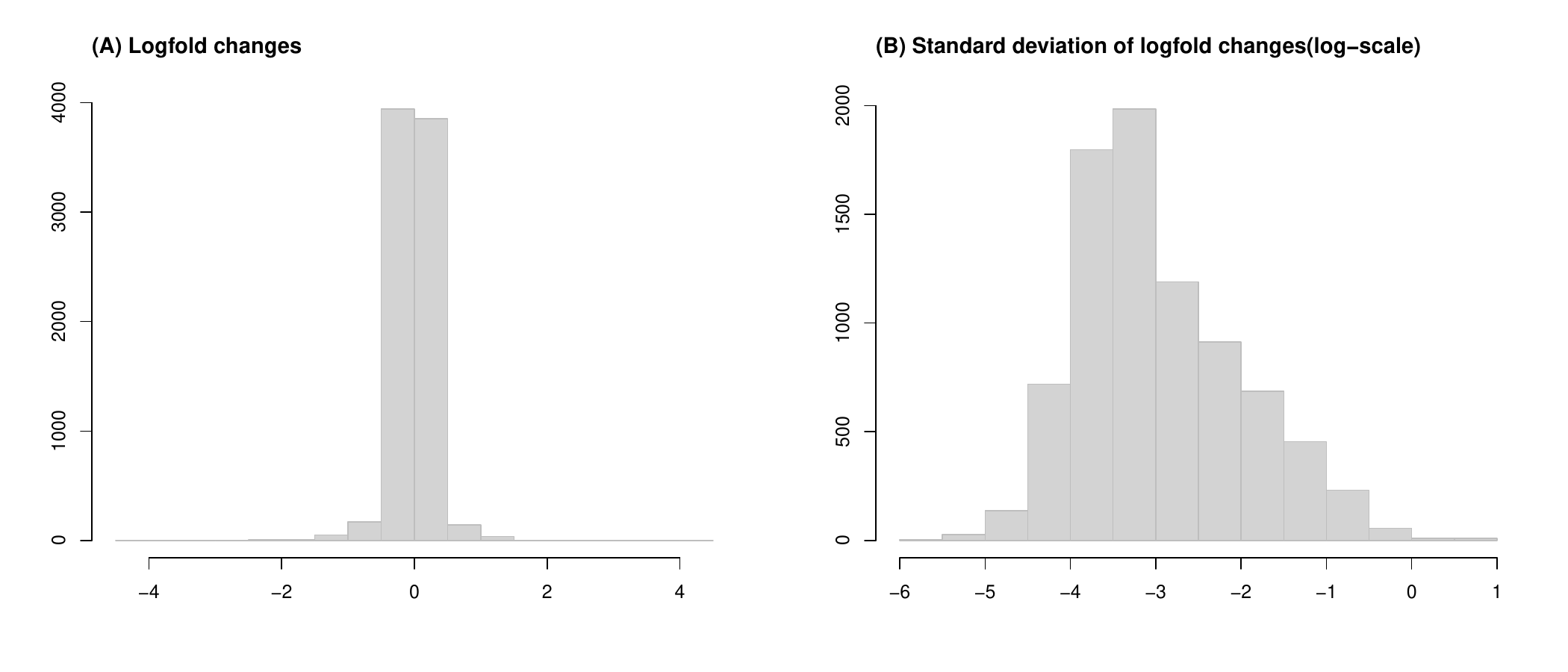}
    \caption{The histograms show that the distributions of (A) the logfold changes and (b) their estimated standard deviations.} 
    \label{fig:hist_margin}
\end{figure}

While the exact distribution of each logfold change remains unknown, it's reasonable to assume that
different sizes of standard deviation of logfold changes
affect the determination of the significance of the logfold change.
Based on this idea, we propose two methods for setting rejection criteria based on standard deviation. The first method adjusts rejection regions according to whether the standard deviation exceeds a certain threshold, similar to a diagnostic screening process. The second method varies the rejection regions with the exact value of the standard deviation, allowing precise adjustments when detailed data on both the primary and auxiliary variables are available.

\section{Models and Estimation of Distributions}
{The traditional multiple testing methods with existing covariate either assume a uniform distribution under the null hypothesis regardless of the covariates or satisfy the mirroring property, where the role of the covariate is to maximise the rejection of hypotheses by adjusting rejection regions for different values of the covariate. More specifically, in the existing studies such as \cite{Lei2018}, \cite{Zhang2019}, \cite{Ignatiadis} and \cite{Boca-Leek}, the null density of $p$ value given the auxiliary variable $y$ is either $I_{[0,1]}(p)$ regardless of $y$ or the mirroring property such as $P( p \leq t |y ) =P( p \geq 1-t|y)$ for $t \in [0,1/2]$. 
\begin{figure}
    \centering
    \includegraphics[width=\linewidth]{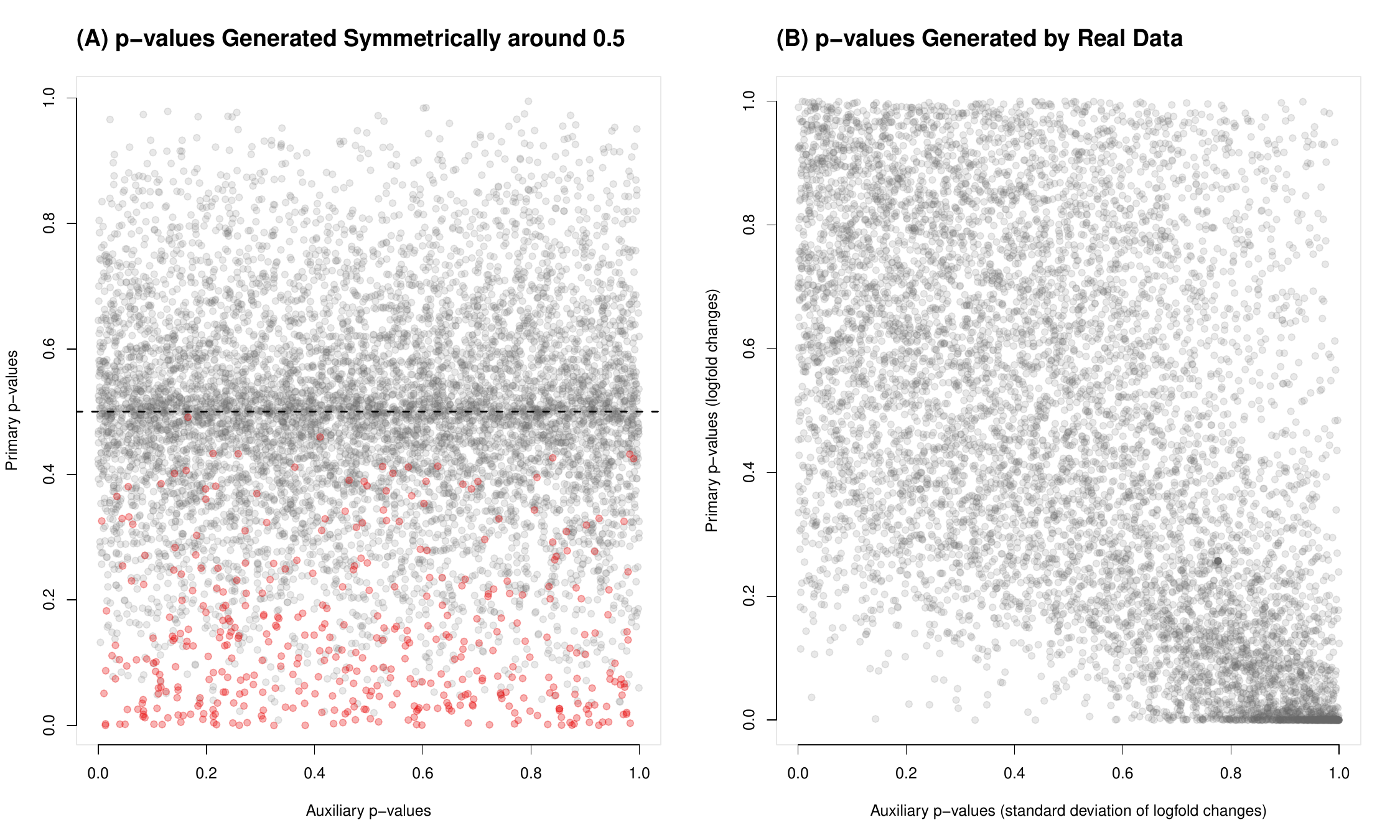}
    \caption{Comparison of $p$-value distributions: (A) $p$-values generated symmetrically around 0.5, and (B) $p$-values generated from Set4$\Delta$ data }
    \label{fig:mirror}
\end{figure}

{It's worth noting that the mirroring property fails in some cases. This property implies a zero correlation between the $p$ value and the auxiliary variable, often indicating a weak relationship between them. Achieving this constraint can be overly restrictive if the auxiliary variable is to be used to increase the power of the test. For example, we generated exemplary $p$-values that are symmetric around 0.5, as shown in Figure \ref{fig:mirror}(A). This distribution of $p$-values is suitable for existing statistical methods that assume such symmetry in their analyses. Conversely, $p$-values generated from real data typically do not exhibit this symmetry around $0.5$, which implies that the underlying assumptions of existing methods are not met.}

Rather than those constraints on the null distribution of $p$ value given the auxiliary variable, our proposed methods are based on considering the joint distribution of the primary variable of interest and an additional auxiliary variable. We estimate the marginal distribution for each variable and derive the joint distribution using the copula method, based on which we propose two methods. In other words, we treat the auxiliary variable as a random variable rather than covariate so that we can consider dependent structure between those two variables. Our proposed methods do not request that the marginal distribution of $p$ value is uniformly distributed or has mirroring property. We use the relationship between the distribution of $p$ value and the auxiliary variable through the joint distribution of them. Therefore, our proposed methods reflect those two variables more flexibly compared to existing methods. We will demonstrate that these proposed methods yield different results compared to the existing methods.}

Suppose we have primary and auxiliary variables denoted by $\beta$ and $y$ respectively. The primary variable, $\beta$, is modelled as a mixture of two probability density functions, with each being assumed to have been generated under null and alternative scenarios, labelled as $f_0$ and $f_1$, respectively. Thus, it is of interest to test $M$ hypotheses of the primary variables,
\begin{eqnarray}
H_{0i} : \beta_i \sim f_0~~ vs.~~H_{1i}: \beta_i \sim f_1
\label{eqn:hypothesis}
\end{eqnarray}
for $1\leq i \leq M$ where each $\beta_i$ is modelled as
\begin{eqnarray}
\label{eqn}
f(\beta) &=&p_0 f_0 (\beta)+ (1 - p_0)f_1 (\beta)
\end{eqnarray}
where $p_0=P(\mbox{$H_{0i}$ is true})$.

This setting has been used in a variety of multiple testing problems. For example, \cite{efron2002empirical} used a normal distribution for $f_0$ of which the mean and standard deviation are estimated based on the zero assumption and the marginal distribution $f$ is estimated using Poisson regression instead of estimating the alternative density $f_1$. In \cite{ramos2021}, the logfold changes statistics for Set4$\Delta$ that come from $f_0$ were modelled as a mixture of two normal distributions rather than a single normal distribution to explain two different sources of the null distribution. Once we estimate $f_0$, we also obtain the cumulative distribution function (cdf), $\hat F_0(\beta) = \int_{-\infty}^\beta \hat f_0(t)dt$.

The use of an auxiliary variable ($y$) is also considered. Since the auxiliary variable is used as supplementary information, we only need to identify the marginal distribution of $y$ such as
\begin{eqnarray}
y \sim h(y)
\end{eqnarray}
where $h(y)$ is the probability density function (pdf) of $y$. In practice, $h(y)$ or the cumulative distribution function $H(y)=\int^{y}_{-\infty} h(t)dt$ can be estimated by nonparametric approaches, for example, we estimate
\begin{eqnarray}
\hat H(y) = \frac{1}{M} \sum_{i=1}^M \mathbf{I}(y_i \leq y)
\end{eqnarray}
which is the empirical cdf of $y_i$s.

This method helps researchers by allowing them to use extra scientific information that isn't the main focus of their study, but still makes their results stronger and more trustworthy. By estimating $\hat H(y)$, the empirical cumulative distribution function of $y$, we apply a non-parametric approach without making restrictive assumptions about its distribution.

Once we have the marginal distributions of the primary and auxiliary variables, we use copula models to estimate the dependence structure of $y$ and $\beta$ of which the joint distribution is $G_0(y,\beta)$. 
The selection of an appropriate copula is a crucial step in modelling the dependence structure in multivariate distributions. There are several approaches to choose an appropriate copula. One ad-hoc approach is based on a scatterplot of the data to gain insight into the dependence structure which can be done by converting the data using the data's marginal cumulative distribution. There are typical patterns of scatterplots for different copulas, so we may choose the closest pattern of the scatterplot. Another approach is comparing Log-Likelihood (LogLik), Akaike Information Criterion (AIC), and Bayesian Information Criterion (BIC). These criteria are widely used for the selection of copulas. For example, see \cite{stober2014regime} and \cite{fink2017regime}. We present examples of copulas in Figure \ref{fig:copula_example} 
and numerical studies of the performances of these criteria for selection of  copula
 in Table \ref{tbl:copula_test} in supplementary material.

In this section, we present our proposed methods depending on how we handle the auxiliary variable. We first define
$\delta_i =
I( {H_{0i} ~\mbox{is rejected}} )
$
and
$\theta_i = {I}(H_{1i} ~\mbox{is true})$
where $I(\cdot)$ is an indicator function. Our proposed methods are designed to control the FDR which is
\begin{eqnarray}
FDR = E \left( \frac{ V}{R \vee 1} \right)
\label{eqn:FDR}
\end{eqnarray}
where $V = \sum_{i=1}^M \delta_i (1-\theta_i)$
is the number of false rejections and $R=\sum_{i=1}^M \delta_i$
is the number of rejections.
We introduce the marginal $p$ values and then propose two types of methods with some theoretical justification.
 
\subsection{Calculation of Marginal $p$-values}
We define $p_{1i}$ and $p_{2i}$
which are the marginal $p$-values corresponding to
the auxiliary variable and the primary variable, respectively.
For observed $y_i$ and $\hat \beta_i$,
\begin{eqnarray}
 p_{1i} &=&  H(y_i), \\
 p_{2i} &=& 2 \min \left( F_0(\hat \beta_i), 1-F_0(\hat \beta_i)\right).  \label{eqn:p2i}
\end{eqnarray}
We mentioned how $f_0$ is estimated in the previous section.
Based on the estimator of $f_0$, we have an estimator of
$F_0(\beta) =\int_{-\infty}^\beta f_0(t) dt$.
We also have an estimator of
$H(y)$ such as the empirical cdf as discussed in the previous section.
We assume that $p_{1i} \sim U[0,1] $ and
$p_{2i} \sim U[0,1]$ under the $H_{0i} : \beta_i \sim f_0$.
Since $(y_i, \hat \beta_i)$ are generated from the joint distribution $G_0$, $p_{1i}$ and $p_{2i}$ are dependent.

\begin{remark2}
Note that
$p_{2i}$ covers the case of two-tailed $p$-value for
non-symmetric distribution.
For one-sided cases, we can define $F_0(\hat \beta_i)$
and $1-F_0(\hat \beta_i)$
for left-tailed and right-tailed cases, respectively.
Similarly, $p_{1i}$ can also be defined for two-tailed cases.
\end{remark2}
Depending on the marginal $p$-values such as $p_{1i}$ and $p_{2i}$, we
propose two methods, namely two-stage FDR(H) and
two-stage FDR(S) in the following section.

\subsection{Two-stage FDR(H)and Two-stage FDR(S)}
As mentioned before,
our proposed methods are based on modelling
the joint distribution of $p_{1i}$ and $p_{2i}$
rather than assuming $p_{1i}$ is $U(0,1)$ regardless of
the auxiliary variable $p_{2i}$
or the mirroring property such as the density of $p_{1i}$
is given $x$ under the null hypothesis is symmetric around $1/2$. We present two proposed methods, called
Two-Stage FDR(H) and Two-Stage FDR(S).

\textbf{Two-stage FDR(H):}
For the first type of two-stage procedure two-stage FDR(H), we proceed as follows:
Given a value of $\gamma_1$, if $p_{1i} > \gamma_1$
equivalently $y_i > H^{-1} (\gamma_{1})$, we conclude that
the first requirement via the auxiliary variable is not satisfied
and hence we do not proceed with testing the significance of the primary variable.
On the other hand, if $p_{1i} \leq \gamma_1$ equivalently $y_i \leq H^{-1} (\gamma_{1})$,
then we test the significance of the primary variable.
Here, $\gamma_1$ is a tuning parameter that serves as a cut-off value to determine whether to proceed with testing
the primary variable.
We will discuss in section \ref{subsec:fdr} how we choose $\gamma_1$ to maximise the number of discoveries.

Considering this two-stage procedure,
we define the $p$-value, $p_i$ aggregating
two $p_{1i}$ and $p_{2i}$ as follows:
for a given $\gamma_1 \in [0,1]$, 
\begin{eqnarray}
     p_i^H(\gamma_1) =
  \begin{cases}
         C(\gamma_1, p_{2i})   & \quad \text{if } p_{1i} \leq \gamma_1 \\
      p_{1i}     & \quad \text{if } p_{1i}>\gamma_1
  \end{cases}
  \label{eqn:p_def1}
\end{eqnarray}
where 
\begin{eqnarray}
C(a,b) = P_{G_0} (P_{1i} \leq a, P_{2i}\leq b)
\label{eqn:C_ab}
\end{eqnarray}
is a 2-dimensional copula distribution evaluated at $(a,b)$. Here, $G_0$ stands for the joint distribution of $y_i's$ and $\beta_i's$ under null hypothesis.
In order to reject the null hypothesis for the primary variable $\beta_i$,
we need two steps such as $p_{1i} \leq \gamma_1$ and a small $p_{2i}$, i.e.,
the event of ${p_{1i} \leq \gamma_1, p_{2i} \leq \gamma_2 }$
for some $\gamma_2$. In testing the primary variable, it is better not to use
$\gamma_2 =\alpha$ since $p_{1i} \leq \gamma_1$ affects the test for the primary variable.
More specifically, take $\gamma_2$ which should satisfy
the uniform distribution of $P_{i}$ such that
$P_{G_0}(P_{1i} \leq \gamma) =P_{G_0}( P_{1i} \leq \gamma_1,
P_{2i} \leq \gamma_2) =\gamma$.
This means $\gamma_2$ depends on $\gamma_1$ and $\gamma$, 
say $\gamma_2 \equiv \gamma_2(\gamma_1,\gamma)$. 
Using the definition of \eqref{eqn:C_ab}, 
if $ \hat \beta_i \geq  D_{0.5}$, then we have 
\begin{eqnarray}
C(\gamma_1, \gamma_2)= P_{ G_0(y,\beta)  }(P_{1i} \leq \gamma_1, P_{2i} \leq \gamma_2) = \int^{C_{1-\gamma_1}}_0 \int_{D_{1-\gamma_2/2}}^{1} G_0(y, \beta) d\beta dy 
= \gamma
\label{eqn:gamma_def}
\end{eqnarray}
where  $C_{\alpha}$ and $D_{\alpha}$ are $\alpha$ quantiles of 
$f_0$ and $h_0$, respectively.; 
if $\hat \beta_i < D_{0.5}  $ 
\begin{eqnarray}
C(\gamma_1, \gamma_2)=P_{ G_0(y,\beta)  }(P_{1i} \leq \gamma_1, P_{2i} \leq \gamma_2) = \int^{C_{1-\gamma_1}}_0 \int^{D_{\gamma_2/2}}_{0} G_0(y, \beta) d\beta dy 
= \gamma
\end{eqnarray}
equivalently, we can simply define $\gamma_2$ as the value which satisfies
\begin{eqnarray}
    C(\gamma_1, \gamma_2) = \gamma
    \label{eqn:gamma_def2}
\end{eqnarray}

This method determines whether to proceed to the next step based on the value of 
$p_{1i}$ in the first step, and if  $p_{1i} >\gamma_1$,  
the corresponding primary variable $y_i$ 
as a whole is not considered significant although the second p-value, $p_{2i}$, is very small.  
In other words, such a thresholding in the first step plays a role of screening genes which do not satisfy the condition based on the auxiliary variable $y_i$.
To highlight this characteristic, we will refer to it as ``two-stage FDR(H)", denoting the implementation of two-stage FDR with hard thresholding.

From two-stage FDR(H), 
it is crucial to check whether $p_i$ has 
a desirable property such that 
$p_i^H$ has the uniform distribution under $G_0$. 
Before we state Theorem \ref{thm:typeH} showing the uniform distribution of $p_i$ in both \eqref{eqn:p_def1} and  
\eqref{eqn:p_def2} under $G_0$, 
we present the following lemma which is used in the proof of 
Theorem \ref{thm:typeH}. 

\begin{lemma2}  
From the definition of $p_i^H(\gamma_1)$ 
in \eqref{eqn:p_def1}, we have 
\begin{eqnarray}
P(P_{1i} \le \gamma_1, P_{2i} \le \gamma_2) = P(P_i^H(\gamma_1)\le \gamma, P_{1i} \le \gamma_1)
\label{eqn:claim}
\end{eqnarray}
for $\gamma=C(\gamma_1,\gamma_2)$ in 
\eqref{eqn:gamma_def2} 
where $\gamma_1$ and $\gamma_2$ are in $[0,1]$.
\label{lemma:pi_1}
\end{lemma2}
\begin{proof}
See the supplementary material. \end{proof}
Using Lemma \ref{lemma:pi_1}, 
we present the following theorem showing 
that $p_i^H(\gamma_1)$ is uniformly distributed regardless of the choice of $\gamma_1$. 

\begin{theorem}  
If $(\hat \beta_i, y_i) \sim G_0(\beta, y)$,  
$p_i^H(\gamma_1)$'s defined in (\ref{eqn:p_def1}) follow the uniform distribution in $(0,1)$.  
\label{thm:typeH}
\end{theorem}
\begin{proof}
See the  supplementary material. \end{proof}
\textbf{Two-stage FDR(S):}
Instead of this stringent thresholding procedure in two-stage FDR(H), 
we also provide the following method, 
two-stage FDR(S). 
The basic idea is to adjust the distribution of the second $p_{2i}$ based on the given value of the first $p_{1i}$, by using the estimated copula. 
In fact, this will be implemented by computing the conditional distribution 
of $p_{2i}$ given $p_{1i}$.
More specifically,  
$p_i^S$ aggregating $p_{1i}$ and $p_{2i}$  is defined as follows:
\begin{eqnarray}
    p_i^S = C(p_{2i}|p_{1i})
    \equiv P_{G_0}(P_{2i} \le p_{2i} | P_{1i} = p_{1i})
\label{eqn:p_def2}
\end{eqnarray}
which is the conditional cumulative distribution of $p_{2i}$ given $p_{1i}$. 

As in two-stage FDR(H), we can show that $p_i^S$ defined in \eqref{eqn:p_def2} also follows uniform distribution under $G_0$.

\begin{theorem} 
If $(\hat \beta_i, y_i) \sim G_0(\beta, y)$,  
$p_i^S$'s defined in (\ref{eqn:p_def2}) follow the uniform distribution in $(0,1)$.  
\label{thm:typeS}
\end{theorem}
\begin{proof}
See the supplementary material. \end{proof}

\textbf{Comparision of Two Methods:}
Two-stage FDR(S)
 differs from two-stage FDR(H) 
 by introducing multiple thresholds. 
 Notably, two-stage FDR(S) does not require a tuning parameter 
such as $\gamma_1$ in two-stage FDR(H). 
Instead of relying on one threshold $\gamma_1$ in two-stage FDR(H), 
two-stage FDR(S) 
uses  $p_{1i}$ to adjust the decision 
rule for $p_{2i}$ through 
the conditional probability. 
We call this process as 
soft thresholds in the sense that 
we adjust rejection region for the primary variable in a smooth way which is demonstrated in Figure \ref{fig:auc}.
To emphasise the contrast with ``two-stage FDR(H)", we refer to this procedure as ``two-stage FDR(S)", representing the concept of a soft margin.

To clarify the differences between two-stage FDR(H) and two-stage FDR(S), 
we present Figure \ref{fig:auc} including (A), (B) and 
(C) for two methods, respectively. 
In each figure,  we consider 
three distinct points, $A$, $B$ and 
$C$ which are corresponding to 
$(p_{1i}, p_{2i}) =(0.4,0.1)$, $(0.8, 0.1)$ 
and $(0.1, 0.143)$. 
In Figure \ref{fig:auc}, 
(A) and (B) are the cases of rejection regions for different $\gamma_1$ values, 
$\gamma_1=0.9$ and $0.7$ respectively.  
We see that the rejection regions 
have rectangular forms which means 
that the cut-off value for $p_{2i}$ is identical whenever $p_{1i} < \gamma_1$.  
From (A) and (B), 
the rejections for $B$ and $C$ are different
 for different $\gamma_1$ values. 
When $\gamma_1=0.9$, all $A$, $B$ and $C$ 
are proceeded to test the primary variable $p_{2i}$, however $C$ is not rejected since $p_{2i}$ is larger than the cut-off value. 
On the other hand, when $\gamma_1=0.7$, 
$B$ does not move to the testing of the primary variable, however $C$ is rejected to be significant. 
This shows that there is trade-off relationship between $\gamma_1$ and the cut-off value for $p_{2i}$. For larger value of $\gamma_1$ (equivalently, weak requirement on the auxiliary variable),  the cut-off value for $p_{2i}$
 is small, i.e., more stringent on testing the primary variable.  Conversely, 
 for small value of $\gamma_1$ (stringent on the auxiliary variable), 
 testing for the primary variable is less stringent, i.e., the cut-off value is larger for $p_{2i}$.  
Regarding two-stage FDR(S), 
we observe similar patterns, however 
the rejection region is not a rectangular form 
which means the rejection for $p_{2i}$ depends on $p_{1i}$ coming from the conditional distribution.

To summarise differences between two methods, 
both methods have a trade-off relationship between the first and second tests, but the major difference lies in the shape of the rejection regions. As mentioned earlier, in the case of two-stage FDR(H), if a variable passes the first stage through $p_{1i} \leq \gamma_1$, 
the same cutoff value is applied to $p_{2i}$ in the second stage. In this case, 
if it was eliminated in the first stage ($p_{1i} >\gamma_1$), there is no possibility of rejecting the primary variable regardless of the value of $p_{2i}$.   On the other hand, in the case of two-stage FDR(S), the rejection region of $p_{2i}$ gradually becomes larger depending on the significance strength shown by the auxiliary variable (small $p_{1i}$). 
Unlike  two-stage 
FDR(H),  two-stage FDR(S) has 
a possibility that the primary variable can be rejected in all range of $p_{1i}$.

\begin{figure}[htbp]
   \centering    
\includegraphics[width = \textwidth]{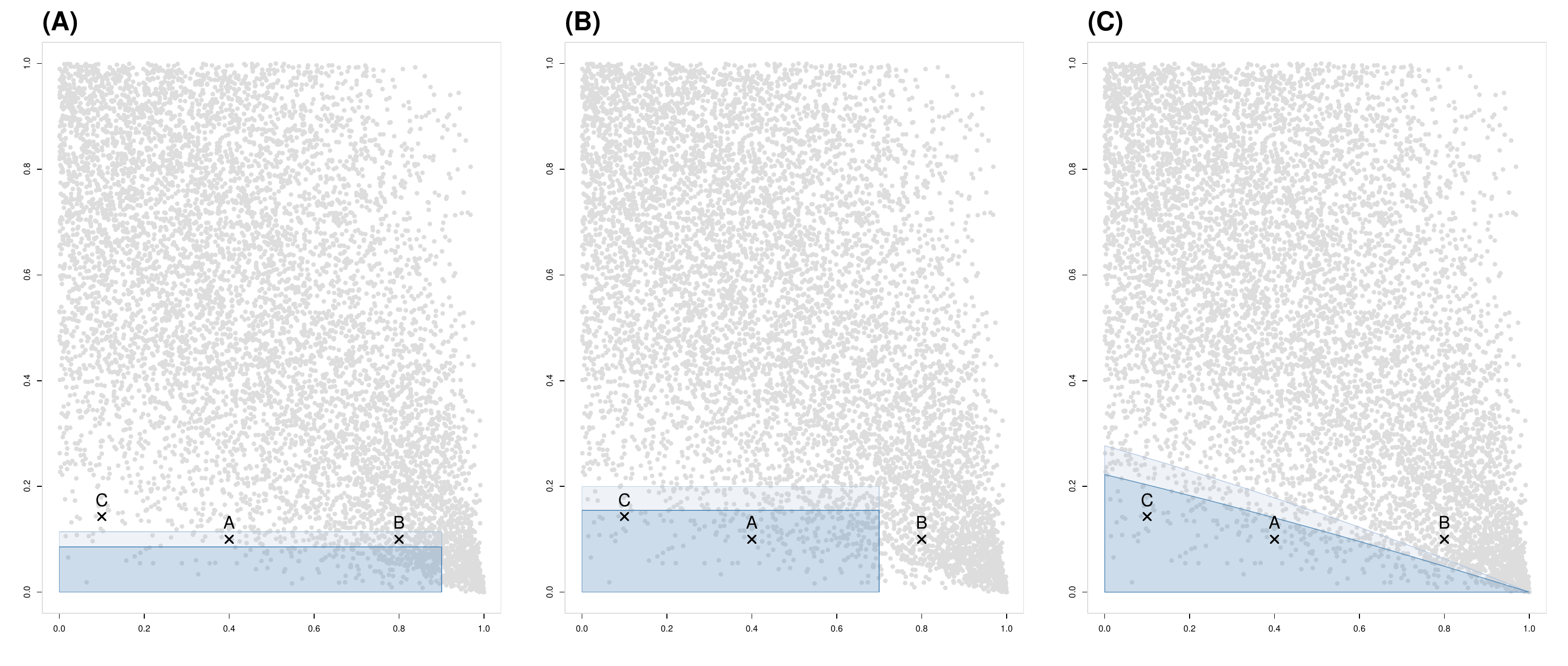}
    \caption{$p_i$'s for two-stage procedures (A),(B) of two-stage FDR(H) and (C) of two-stage FDR(S). In (A), $p_i^H$ with the threshold $\gamma_1=0.9$ and (B) with the threshold $\gamma_1=0.7$}
     \label{fig:auc}
\end{figure}

\subsection{False discovery rate controlling procedures}
\label{subsec:fdr}
In the previous section, 
Theorem \ref{thm:typeH} and \ref{thm:typeS} show that suggested $p$-values for both methods 
are uniformly distributed when $(\hat \beta_i, y_i) \sim G_0(\beta, y)$. 

Based on this property of $p_i$s,  we apply a method to control the FDR using $p_i^H$s and $p_i^S$s 
in \eqref{eqn:p_def1} and \eqref{eqn:p_def2}. Our goal is to estimate the FDR and ensure that its value does not exceed a given FDR level, denoted as $\alpha$.
We test $M$ hypotheses  in \eqref{eqn:hypothesis}, 
for example, in the Set4$\Delta$ data  introduced in section \ref{sec:data}, we test 
the significance of logfold changes .  

For a given $\gamma$, the estimated FDR is expressed as:
\begin{equation}
\label{eqn:storey1}
\widehat{FDR ({\gamma})}=\frac{\widehat{\pi_0}\gamma M }{\max( R_M(\gamma ),1)} 
\end{equation}
where 
\begin{itemize}
\item $M$ represents the number of hypotheses tested, 
\item $\widehat{\pi_0}$ is an estimate of the proportion of true negatives among $M$,
\item $R_M(\gamma) =\#\{i : p_i \leq \gamma, 1\leq i \leq M \}$ is the number of rejections, where $p_i$ is either $p_i^H(\gamma_1)$ or $p_i^S$. 
\end{itemize}
The denominator in \eqref{eqn:storey1} takes the maximum value between the number of rejected hypotheses and 1. 
Here,
$\widehat{\pi_0}$ is estimated by using 
\begin{equation}
\widehat{\pi_0}=\frac{\#\{p_i>\lambda\}}{(1-\lambda)M}
\label{eqn:pi0}
\end{equation}
and we set $\lambda=0.5$ for this study.

\textbf{Selection of $\gamma_1$ in two-stage FDR(H)} \\
Since $p_i^H(\gamma_1)$ is uniformly distributed with any $\gamma_1\in[0,1]$, we select $\gamma_1$ which rejects most to maximise the true positive rate. i.e. 
\begin{align}
    \widehat{\gamma} &= \argmax_{\gamma} \big(\widehat{FDR ({\gamma})} \leq \alpha \big), \label{eqn:gamma}\\
    \widehat{\gamma_1} &= \argmax_{\gamma_1}\big(\#\{i:p_i^H(\gamma_1) \leq \widehat{\gamma}, 1 \leq i \leq M\}\big),
    \label{eqn:hatgamma1}\\
    p_i^H &= p_i^H(\hat{\gamma_1}).
\end{align}

We summarise our proposed methods
controlling FDR in {\bf Algorithm \ref{alg:typeH}} and {\bf Algorithm \ref{alg:typeS}}.

\begin{algorithm}[ht]
\caption{Two-stage FDR(H)}
\label{alg:typeH}
\begin{algorithmic}[1]
\State Estimate null distribution of $\beta$
\State Compute $p_{2i}$ under the null distribution of $\beta$, and $p_{1i}$ using the empirical distribution of $Y$
\State Select the copula between $p_{1i}$ and $p_{2i}$
\For{$\gamma_1 \in \{r_1, r_2, \ldots, r_K\} \subset [0,1]$}
    \State Compute $p_i^H(r_k)$ using $p_{1i}$, $p_{2i}$
    \State Compute $\widehat{\pi_0}$
    \State For fixed values of $\alpha$ and $\lambda$, check the number of rejections ($m_k$)
\EndFor
\State Compare the number of rejections for each $r_k$ and choose $\widehat{\gamma_1} = r_k$ as the one that maximises $m_k$
\State Reject if $P_i^H(\widehat{\gamma_1})\leq \widehat{\gamma}$, where $\gamma$ is defined in \eqref{eqn:gamma} with $\widehat{\gamma_1}$
\end{algorithmic}
\end{algorithm}

\begin{algorithm}[ht]
\caption{Two-stage FDR(S)}
\label{alg:typeS}
\begin{algorithmic}[1]
\State Estimate null distribution of $\beta$
\State Compute $p_{2i}$ under the null distribution of $\beta$, and $p_{1i}$ using the empirical distribution of $Y$
\State Select the copula between $p_{1i}$ and $p_{2i}$
\State Compute $p_i^S$ using $p_{1i}$, $p_{2i}$ with estimated copula
\State Compute $\widehat{\pi_0}$
\State Estimate $\widehat{\gamma}$ satisfying \eqref{eqn:gamma}
\State Reject if $p_i^S \leq \widehat{\gamma}$
\end{algorithmic}
\end{algorithm}

\section{Simulation Studies}
In our study, we generate a dataset of $M=8,000$ pairs of variables, $(\beta_i, y_i)$, for $1 \leq i \leq 8,000$. This dataset is designed to simulate the characteristics of our real data example, the Set4$\Delta$ data, but in a simplified form. 
The auxiliary variable $y_i$ 
\begin{eqnarray}
y_i \sim Gamma (3,4)
\end{eqnarray}
with the mean value $3/4$ 
and the primary variable $\beta_i$
\begin{eqnarray*}
\beta_i \sim f(\beta)  &=& 
p_0f_0(\beta) + (1-p_0) f_1(\beta)\\
&=& 0.95 \phi(\beta) + 0.05(0.5 \phi(\beta-\mu) + 0.5\phi(\beta+\mu))
\end{eqnarray*}
for some $\mu>0$. 
To closely replicate the joint distribution observed in the Set4$\Delta$ data, we model the pairs $(\beta_i, y_i)$ to be jointly distributed following a Clayton copula.

In our first scenario we used the Clayton copula to calculate \(p\) values. We varied the parameters \(\mu\) and Kendall's \(\tau\) and assessed the \({FDR}\) in \eqref{eqn:FDR}. 
Our Monte Carlo simulations provide an approximate FDR which is  
$FDR \approx \frac{1}{K} \sum_{k=1}^K  \frac{ V_k }{R_k \vee 1}$ 
where $R_k$ and $V_k$ are 
the number of rejections and that of false rejections, respectively at $k$th simulation among $K$ simulations. 
Let $\theta_i = {I}(H_{1i} ~\mbox{is true})$.  
As statistical powers,  we use the true positive rate (TPR) which is 
\begin{eqnarray}
 TPR = E\left( \frac{ \sum_{i=1}^M \theta_i \delta_i }{\sum_{i=1}^M\theta_i \vee 1}\right)  
 \label{eqn:TPR}
\end{eqnarray}
where  $a\vee 1 =\max(a,1)$ is the maximum of $a$ and 1,    $\sum_{i=1}^M \theta_i $ 
is the number of $H_{1i}$ and 
$\sum_{i=1}^M \theta_i \delta_i $ 
is the number of the true positives. 
Similarly, an approximate 
TPR is calculated  $\frac{1}{K}\sum_{k=1}^K \frac{S_k}{M_k}$ 
where $S_k$  and $M_k$ 
are  the number of true rejections and that of $H_{1i}$  
at $k$th simulation among $K$ simulations. We use $K=1000$ in our simulation studies.   
{We consider the various methods as follows : 
locfdr in \cite{efron2004large}, Storey in \cite{storey2002}, IHW in \cite{Ignatiadis}, 
Boca and Leek in \cite{Boca-Leek} and 
our proposed two-stage FDR methods such as Type H and Type S.  }

{
The results shown in Tables \ref{tbl:sim_tau} show the effects of different degrees of dependence (\( \tau \)) with a constant distance between the null and alternative distributions (\( \mu = 3.0 \)), and \( \mu \) changes while \( \tau \) remains fixed at -0.4. Notably, all methods indicate an increase in TPR as \( \mu \) increases.
The one-stage FDR methods, which depend only on \( \hat{\beta}_i \) values, show consistent FDR and TPR, highlighting their limitations in adapting to different levels of dependence. Conversely, both covariate-assisted and two-stage FDR methods show higher TPR with increasing dependence. However, their effectiveness in controlling FDR varies. Covariate-assisted methods, despite extending the rejection region to reject more hypotheses, fail to maintain FDR control when dependencies are present. 
This suggests that covariate-assisted methods extend the rejection region where null hypotheses accumulate. 
Meanwhile, two-stage FDR methods effectively maintain FDR control even as rejection rates increase with increasing dependency, thereby achieving a better TPR.\\
This result suggests that the inclusion of auxiliary variables can improve the discovery of significant results, but the results should be interpreted with caution. In Section 5, we demonstrate the potential for increased discovery using auxiliary variables. Still, the results of covariate-assisted methods should be treated with particular caution.
}

\begin{table}[]
\centering
\caption{Comparison of False Discovery Rate (FDR) and True Positive Rate (TPR) across different settings for the mean parameter, $\mu$, and Kendall's $\tau$ coefficient.}
\label{tbl:sim_tau}
\resizebox{\textwidth}{!}{%
\begin{tabular}{ccccccccc}
\hline
 & \multirow{2}{*}{$\tau$} & \multicolumn{2}{c}{One-stage FDR} & \multicolumn{3}{c}{Covariate-assisted FDR} & \multicolumn{2}{c}{Two-stage FDR} \\ \cline{3-9} 
 &  & locfdr & Storey & IHW & Boca and Leek & AdaFDR & type H & type S \\ \hline
\multirow{5}{*}{$\widehat{FDR}$} & 0.0 & \multirow{5}{*}{0.013 (0.013)} & \multirow{5}{*}{0.046 (0.020)} & 0.043 (0.020) & 0.000 (0.006) & 0.048 (0.024) & 0.046 (0.021) & 0.045 (0.020) \\
 & -0.2 &  &  & 0.063 (0.032) & 0.000 (0.006) & 0.160 (0.061) & 0.044 (0.019) & 0.038 (0.016) \\
 & -0.4 &  &  & 0.070 (0.039) & 0.000 (0.007) & 0.480 (0.085) & 0.038 (0.016) & 0.028 (0.012) \\
 & -0.6 &  &  & 0.074 (0.041) & 0.014 (0.020) & 0.837 (0.019) & 0.032 (0.018) & 0.017 (0.009) \\
 & -0.8 &  &  & 0.063 (0.045) & 0.945 (0.002) & 0.873 (0.007) & 0.029 (0.028) & 0.005 (0.006) \\ \hline
\multirow{5}{*}{$\widehat{TPR}$} & 0.0 & \multirow{5}{*}{0.213 (0.038)} & \multirow{5}{*}{0.378 (0.051)} & 0.367 (0.052) & 0.002 (0.006) & 0.375 (0.059) & 0.379 (0.051) & 0.376 (0.051) \\
 & -0.2 &  &  & 0.364 (0.050) & 0.004 (0.008) & 0.380 (0.073) & 0.490 (0.042) & 0.577 (0.042) \\
 & -0.4 &  &  & 0.364 (0.047) & 0.014 (0.020) & 0.500 (0.062) & 0.643 (0.033) & 0.804 (0.026) \\
 & -0.6 &  &  & 0.392 (0.039) & 0.115 (0.052) & 0.656 (0.037) & 0.739 (0.025) & 0.928 (0.014) \\
 & -0.8 &  &  & 0.466 (0.039) & 0.582 (0.028) & 0.672 (0.051) & 0.790 (0.022) & 0.980 (0.007) \\ \hline
\end{tabular}%
}
\label{tbl:sim_mu}
\resizebox{\textwidth}{!}{%
\begin{tabular}{cccclllcc}
\hline
 & \multirow{2}{*}{$\mu$} & \multicolumn{2}{c}{One-stage FDR} & \multicolumn{3}{l}{Covariate-assisted FDR} & \multicolumn{2}{c}{Two-stage FDR} \\ \cline{3-9} 
 &  & locfdr & Storey & IHW & Boca and Leek & AdaFDR & type H & type S \\ \hline
\multirow{5}{*}{$\widehat{FDR}$} & 2.0 & 0.022 (0.079) & 0.026 (0.065) & 0.174 (0.084) & 0.000 (0.000) & 0.643 (0.247) & 0.030 (0.022) & 0.024 (0.016) \\
 & 2.5 & 0.016 (0.026) & 0.038 (0.030) & 0.121 (0.057) & 0.000 (0.011) & 0.691 (0.101) & 0.033 (0.018) & 0.026 (0.012) \\
 & 3.0 & 0.013 (0.013) & 0.046 (0.020) & 0.070 (0.039) & 0.000 (0.007) & 0.580 (0.085) & 0.038 (0.016) & 0.028 (0.012) \\
 & 3.5 & 0.010 (0.008) & 0.049 (0.017) & 0.081 (0.047) & 0.003 (0.007) & 0.442 (0.088) & 0.045 (0.018) & 0.030 (0.012) \\
 & 4.0 & 0.008 (0.006) & 0.051 (0.016) & 0.130 (0.053) & 0.008 (0.008) & 0.404 (0.080) & 0.050 (0.018) & 0.031 (0.012) \\ \hline
\multirow{5}{*}{$\widehat{TPR}$} & 2.0 & 0.010 (0.007) & 0.015 (0.013) & 0.073 (0.026) & 0.000 (0.000) & 0.183 (0.090) & 0.190 (0.035) & 0.308 (0.045) \\
 & 2.5 & 0.065 (0.022) & 0.127 (0.041) & 0.168 (0.034) & 0.001 (0.002) & 0.348 (0.070) & 0.420 (0.039) & 0.595 (0.038) \\
 & 3.0 & 0.213 (0.038) & 0.378 (0.051) & 0.364 (0.047) & 0.014 (0.020) & 0.500 (0.062) & 0.643 (0.033) & 0.804 (0.026) \\
 & 3.5 & 0.428 (0.044) & 0.642 (0.041) & 0.614 (0.037) & 0.175 (0.060) & 0.681 (0.055) & 0.803 (0.024) & 0.920 (0.016) \\
 & 4.0 & 0.643 (0.038) & 0.828 (0.027) & 0.781 (0.026) & 0.454 (0.063) & 0.838 (0.035) & 0.902 (0.017) & 0.973 (0.009) \\ \hline
\end{tabular}%
}
\end{table}

Our simulation results indicate that our proposed methods can control the False Discovery Rate (FDR) when the dependency structure is known. However, in real data, the copula is often unknown. To address the question, "What happens if we mis-specify the copula in the models?" we explore the consequences of choosing an incorrect copula for calculating $p_i$ in our proposed methods. While the correct copula selection is essential for controlling the FDR, it is equally crucial to understand the consequences of using an incorrect copula. To evaluate the robustness of our methodology to the selection of copula, 
we choose a copula that does 
not align with the true underlying dependencies. Specifically, we chose the BB7, BB6 and Joe copulas for comparison. These were chosen because of their similarity to the Clayton copula, which allowed us to explore the effects of potential misspecification when replacing the Clayton copula with these alternatives. Table \ref{tbl:sim_cop} displays the estimated FDR and true positive rate (TPR) for each of the four copulas when $\mu=3$ and $\tau=-0.4$.

\begin{table}
\centering
\caption{Comparison of False Discovery Rate (FDR) and True Positive Rate (TPR) in the Presence of Misspecified Copulas.}

\label{tbl:sim_cop}
\begin{tabular}{cccccc}
\hline
                                 & \multirow{2}{*}{Copula} & \multicolumn{2}{c}{One-stage FDR}                                 & \multicolumn{2}{c}{Two-stage FDR} \\ \cline{3-6} 
                                 &                         & locfdr                         & Storey                         & type H         & type S         \\ \hline
\multirow{4}{*}{$\widehat{FDR}$} & Clayton                 & \multirow{4}{*}{0.013 (0.013)} & \multirow{4}{*}{0.046 (0.020)} & 0.038 (0.016)  & 0.028 (0.012)  \\
                                 & BB7                     &                                &                                & 0.006 (0.012)  & 0.001 (0.0.02) \\
                                 & BB6                     &                                &                                & 0.008 (0.009)  & 0.002 (0.003)  \\
                                 & Joe                     &                                &                                & 0.049 (0.019)  & 0.032 (0.014)  \\ \hline
\multirow{4}{*}{$\widehat{TPR}$} & Clayton                 & \multirow{4}{*}{0.213 (0.038)} & \multirow{4}{*}{0.378 (0.051)} & 0.643 (0.033)  & 0.804 (0.026)  \\
                                 & BB7                     &                                &                                & 0.466 (0.035)  & 0.588 (0.030)  \\
                                 & BB6                     &                                &                                & 0.518 (0.037)  & 0.630 (0.032)  \\
                                 & Joe                     &                                &                                & 0.660 (0.032)  & 0.805 (0.026)  \\ \hline
\end{tabular}
\begin{tabular}{cccccc}
\hline
                                 & \multirow{2}{*}{Copula} & \multicolumn{2}{c}{One-stage FDR}                                 & \multicolumn{2}{c}{Two-stage FDR} \\ \cline{3-6} 
                                 &                         & locfdr                         & Storey                         & type H         & type S         \\ \hline
\multirow{4}{*}{$\widehat{FDR}$} & Clayton                 & \multirow{4}{*}{0.013 (0.013)} & \multirow{4}{*}{0.046 (0.020)} & 0.038 (0.016)  & 0.028 (0.012)  \\
                                 & BB7                     &                                &                                & 0.038 (0.017)  & 0.029 (0.012)  \\
                                 & BB6                     &                                &                                & 0.036 (0.016)  & 0.034 (0.013)  \\
                                 & Joe                     &                                &                                & 0.031 (0.014)  & 0.028 (0.011)  \\ \hline
\multirow{4}{*}{$\widehat{TPR}$} & Clayton                 & \multirow{4}{*}{0.213 (0.038)} & \multirow{4}{*}{0.378 (0.051)} & 0.643 (0.033)  & 0.804 (0.026)  \\
                                 & BB7                     &                                &                                & 0.643 (0.033)  & 0.805 (0.026)  \\
                                 & BB6                     &                                &                                & 0.641 (0.033)  & 0.800 (0.026)  \\
                                 & Joe                     &                                &                                & 0.643 (0.033)  & 0.804 (0.026)  \\ \hline
\end{tabular}
\end{table}

From Table \ref{tbl:sim_cop},   
we see that 
all four methods control a given level of FDR for all settings while locfdr and two-stage FDR (S) obtain lower FDR values compared to the nominal level $\alpha=0.05$. 
However, two-stage FDR (S) and two-stage FDR (H) obtain 
highest  and the second highest TPRs, respectively across all settings.  These results show that 
the supplementary information from the auxiliary variable may affect the power of testing procedures when it is appropriately reflected in testing procedure.

\section{Real Data Analysis}
\subsection{Estimation of the Marginal $p$-values and Optimal Copula}

\textbf{auxiliary variable (\(p_{1i}\)) : } The standard deviation of the logfold changes of each gene is used as an auxiliary variable. Since we have small number of samples, we estimated the standard deviation by bootstrap method for all possible combinations, resulting in 729 (i.e. $3^3\times3^3$) possible combinations, and derived the $p_{1i}$ values from their empirical distribution. 
The detailed proceudure is described in {\bf Algorithm} \ref{alg:lfc_sd}.

\begin{algorithm}
\caption{Estimation of Standard Deviation using Bootstrap for Set4$\Delta$ data}
\label{alg:lfc_sd}
\small
\begin{algorithmic}[1]
\For{$i \gets 1$ to $8220$} \Comment{\parbox[t]{0.5\linewidth}{Iterate through all genes}}
    \For{$b \gets 1$ to $729$} \Comment{\parbox[t]{0.5\linewidth}{Iterate through all bootstrap samples}}
        \State \parbox[t]{0.9\linewidth}{Sample three knockouts ($KO_{i(1)}^b, KO_{i(2)}^b, KO_{i(3)}^b$) from the original knockouts ($KO_{i1}, KO_{i2}, KO_{i3}$) with replacement.}
        \State \parbox[t]{0.9\linewidth}{Sample three wildtypes ($WT_{i(1)}^b, WT_{i(2)}^b, WT_{i(3)}^b$) from the original wildtypes ($WT_{i1}, WT_{i2}, WT_{i3}$) with replacement.}
        \State \parbox[t]{0.9\linewidth}{Calculate logfold changes: $\Tilde{\beta}_{i}^b = \log \frac{\bar{KO}^b}{\bar{WT}^b}$}
    \EndFor
    \State Compute the sample standard deviation of $\Tilde{\beta}_{i}^b$ values
    \State Determine the $p_{1i}$ value based on the empirical distribution of these standard deviations
\EndFor
\end{algorithmic}
\end{algorithm}
\textbf{primary variable (\(p_{2i}\)) : } For the primary variable, we adopt the approach proposed by \cite{ramos2021} to evaluate the log change values of each gene. This method estimates the null distribution as a mixture of two normal distributions.i.e. 
\begin{equation*}
    f_0(\beta) = 0.615\phi\big(\frac{\beta}{0.063}\big) + 0.385\phi\big(\frac{\beta+0.002}{0.205}\big)
\end{equation*} 
where $\phi$ is a probability distribution function of a normal distribution.
Based on this null, $p_{2i}$ is calculated as described in \eqref{eqn:p2i}. 

\textbf{Optimal Copula :} we evaluated the validity of various copulas - including Gaussian, Frank, Clayton, Gumbel and Joe - for modelling the dependencies between $p_{1i}$ and $p_{2i}$. This evaluation, based on Log-Likelihood (LogLik), Akaike Information Criterion (AIC) and Bayesian Information Criterion (BIC), is illustrated in Table \ref{tbl:copula_real}. The Clayton copula was selected owing to its outstanding performance in all criterion, indicating an optimal balance between model complexity and fit.

\begin{table}
\centering
\caption{The Log-likelihood(LogLik), AIC, BIC of Gaussian, Frank, Clayton, Gumbel and Joe copulas with set4$\Delta$ dataset}
\label{tbl:copula_real}
\begin{tabular}{@{}cccc@{}}
\toprule
Family   & LogLik          & AIC              & BIC              \\ \midrule
Gaussian & 306.45          & -610.91          & -604.80          \\
Frank    & 332.84          & -663.68          & -657.57          \\
Clayton  & \textbf{393.08} & \textbf{-784.16} & \textbf{-778.05} \\
Gumbel   & 201.79          & -401.57          & -395.46          \\
Joe      & 96.29           & -190.57          & -184.46          \\ \bottomrule
\end{tabular}%
\end{table}

\subsection{Application of FDR Controlling Procedure}

Based on the estimation of the distribution $G_0$ in the previous section, we apply the {\bf Algorithm \ref{alg:typeH}} and {\bf Algorithm \ref{alg:typeS}} with $\alpha = 0.10$. {For the distributions of $p_i^H$ and $p_i^S$, see Figure \ref{fig:hist_pval} in the supplementary material.} 
In \textbf{Algorithm \ref{alg:typeH}}, the process of selecting $\gamma_1$ involved analysing the relationship between the number of hypothesis rejections and various $\gamma_1$ values. Using the equation \eqref{eqn:hatgamma1}, we determined the optimal value of $\hat{\gamma}_1$ as the value that maximises the number of rejections. Figure \ref{fig:gamma_selection} illustrates the relationship between $\gamma_1$ values and the number of rejections. It shows that the number of rejections increases as $\gamma_1$ increases, peaking at $\gamma_1 = 0.987$, after which it starts to decrease. With the choice of $\hat{\gamma}_1 = 0.987$, this resulted in the rejection of $492$ null hypotheses out of a total of $8,220$.

\begin{figure}[ht]
    \centering
    \includegraphics[width=0.7\textwidth]{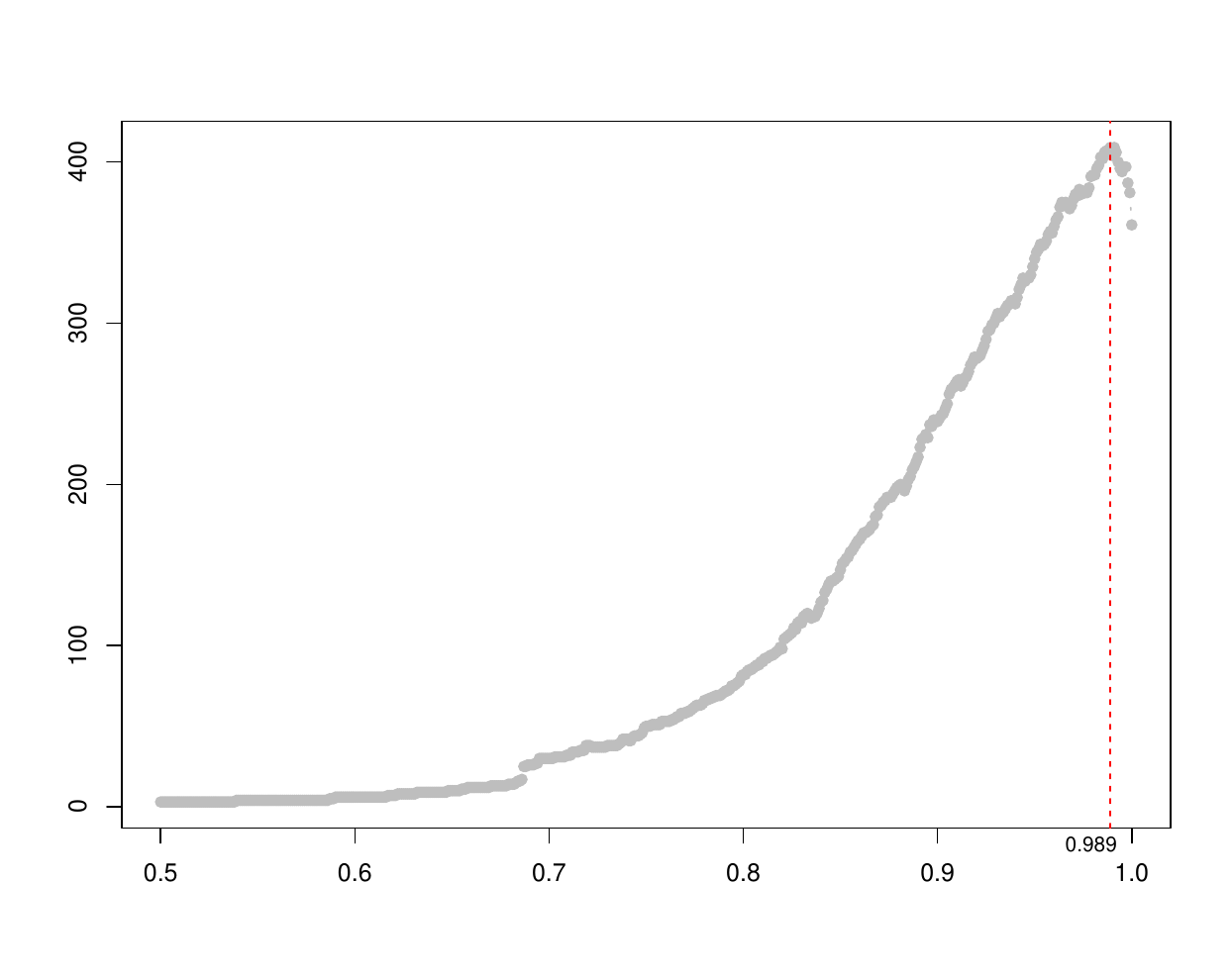}
    \caption{Selection of $\gamma$ in two-stage FDR(H). The x-axis displays $\gamma$ and the y-axis represents the corresponding number of rejection with two-stage FDR(H).}
    \label{fig:gamma_selection}
\end{figure}

In order to assess the performance of four different FDR controlling methods, we compare the number of rejections for both one-stage and two-stage approaches. Firstly, employing the one-stage FDR(locfdr), we identified $249$ rejected genes. A slightly higher number of genes, $424$, were rejected when applying the one-stage FDR(Storey). However, when utilizing two-stage FDR methods, we observed a significant increase in gene rejections. Specifically, the two-stage FDR(H) resulted in $485$ rejected genes, while the two-stage FDR(S) led to a $582$ rejections. These findings support the simulation results presented in Table \ref{tbl:sim_tau}, demonstrating the greater power of two-stage FDR methods in detecting deferentially expressed genes (DEG). Figure \ref{fig:venn} in online supplementary matierial offers an in-depth comparative analysis of the proposed two-stage FDR methods introduced in this study versus traditional one-stage FDR methods (local FDR and Storey's method), applied to the set4$\Delta$ mutant dataset.

Figure \ref{fig:rej_region} shows  
four different rejection regions corresponding to the two one-stage FDR methods and the two two-stage FDR methods.  
These different rejection regions
provide further insight into the comparison between one-stage and two-stage FDR methods. 
As shown in the figures, the one-stage FDR methods determine the rejection region based on only  the logfold changes in the $y$-axis, while in the case of the two-stage FDR(H), it exhibits a rectangular-shaped rejection region determined independently by cutoff values for $y$ and $\beta$. On the other hand, in the case of two-stage FDR(S), it shows a curved rejection region, indicating that the dependence between the two variables has been taken into account.
{The further result compared to DESeq2 (\cite{love2014moderated}), a well-known DEG analysis, is also represented in the supplementary material.}

{\bf HXT gene family :}  We conducted an investigation into the functional role of these newly rejected gene sets using GO and MIPs databases  of the two rejected gene lists in \citet{robinson2002}.
In both cases, we observed an enrichment of genes associated with stress response and cell wall maintenance among the differentially expressed genes, consistent with our previous findings in \cite{jethmalani2021}. Additionally, the new analysis revealed an interesting enrichment of genes involved in hexose transport, particularly from the HXT gene family.
Figure \ref{fig:rej_region} include 
the locations of the HXT gene family represented by crosses and 
show that the two-stage FDR(S) reject all of them  while the others reject part of them.
These rejected genes mainly belong to the minor category of hexose transporters, known for primarily transporting mannitol and sorbitol (\citet{diderich1999}). Remarkably, under stress conditions like hypoxia or anaerobia, the transcription of these transporters is altered, likely leading to increased transport of the preferred carbon source, glucose, into cells (\citet{rintala2008}). This finding identifies another class of genes regulated by Set4 during hypoxia that play a crucial role in the cell's adaptation to low oxygen. Moreover, a number of these genes are located within approximately 40 kB of the chromosome end, including HXT15, HXT13, and HXT8, consistent with previous findings that Set4 regulates genes within the subtelomeric regions of chromosomes (\citet{jethmalani2021}). In conclusion, this differential gene expression analysis method has successfully identified a new class of genes involved in stress responses, regulated by the Set4 protein in yeast grown under hypoxic conditions. 

\begin{figure}[htbp]
    \centering
    \includegraphics[width=\textwidth]{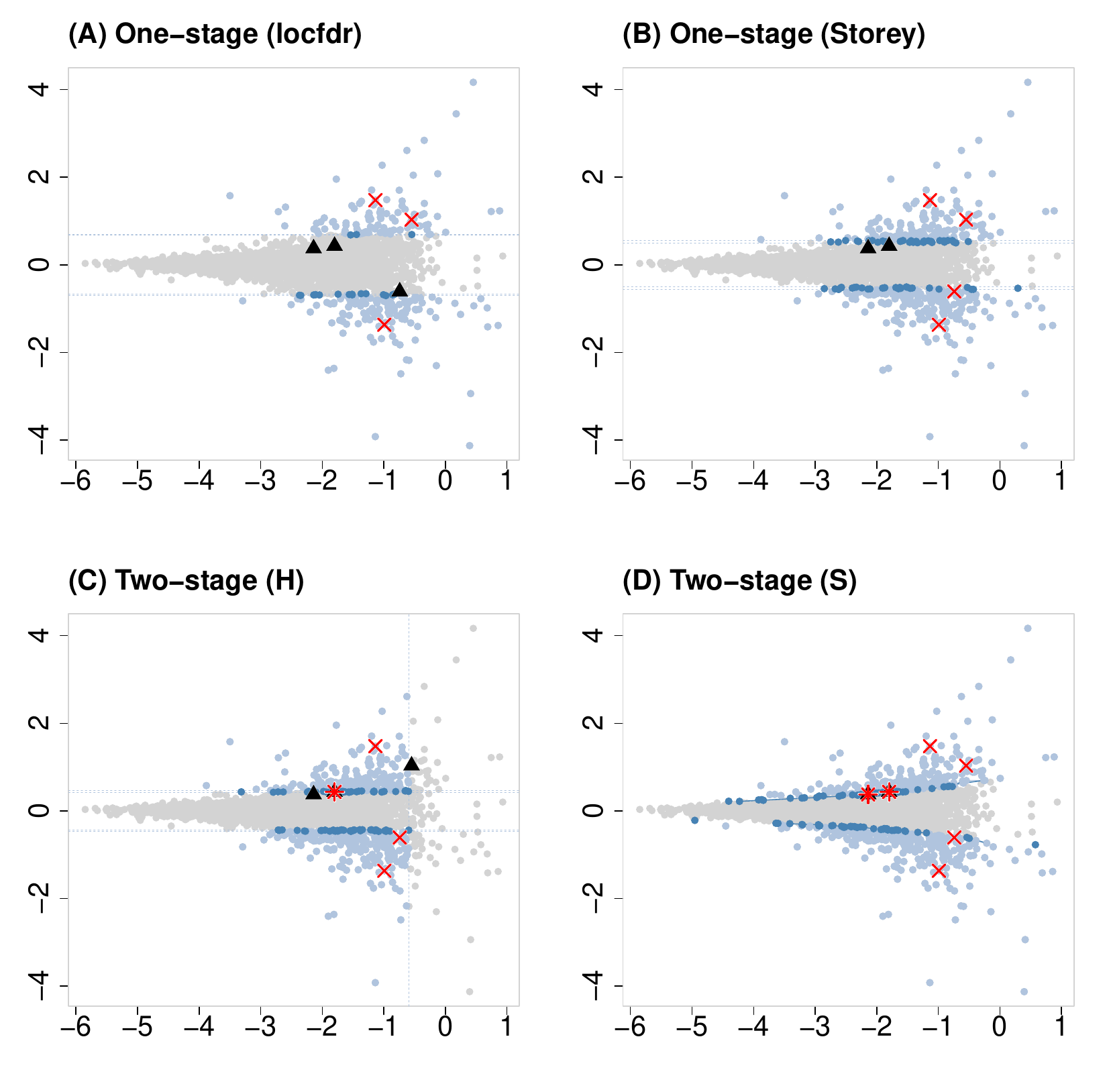}
    \caption{Comparison of Rejection Regions: One-Stage vs. Two-Stage FDR Methods. The x-axis represents the standard deviation of logfold changes on a logarithmic scale as an auxiliary variable, while the y-axis shows the logfold changes as the primary variable. The sky-blue points corresponds to the rejection region at $\alpha < 0.05$, and the darker blue points represents the rejection region at $\alpha < 0.10$. In the legend, \textcolor{red}{\cross} and \textcolor{red}{\rstar} signify the HXT gene family rejection with $\alpha = 0.05$ and $\alpha = 0.10$, respectively, while \tri indicates non-rejected genes in the HXT gene family.}
    \label{fig:rej_region}
\end{figure}

\begin{table}[]
\centering
\caption{Comparison of total and HXT family gene rejections using one-stage (locfdr, Storey), covarate-assisted (IHW, Boca and Leek, AdaFDR) and two-stage (types H and S) FDR methods at $\alpha = 0.05$ and $\alpha = 0.10$.}
\label{tbl:nrej}
\resizebox{\textwidth}{!}{%
\begin{tabular}{ccccccccc}
\hline
 &  & \multicolumn{2}{c}{One-stage FDR} & \multicolumn{3}{c}{Covariate-assisted FDR} & \multicolumn{2}{c}{Two-stage FDR} \\ \cline{3-9} 
 & $\alpha$ & locfdr & Storey & IHW & Boca and Leek & AdaFDR & type H & type S \\ \hline
\multirow{2}{*}{\# rejected genes} & 0.05 & 230 & 361 & 486 & 110 & 844 & 407 & 503 \\
 & 0.10 & 249 & 424 & 633 & 116 & 1329 & 485 & 582 \\
\multirow{2}{*}{\# of rejected HXT family} & 0.05 & 3 & 4 & 5 & 3 & 6 & 3 & 4 \\
 & 0.10 & 3 & 4 & 5 & 3 & 6 & 4 & 6 \\ \hline
\end{tabular}%
}
\end{table}

\section{Concluding Remarks}
In this paper, we develop the two stage procedures controlling FDR based on 
incorporating an auxiliary variable to improve the power of testing the primary variable. Furthermore, we also  increase the credibility of the results via incorporating prior knowledge through the auxiliary variable. 
Rather than assuming the structure of standard deviation of logfold changes, we simply used the copula to model the dependence between the logfold changes and their standard deviations. Our methods allow for a more flexible and accurate representation of the joint distribution of these two variables. 
Two proposed methods select the significant genes in different ways. One takes the form of screening, while the other manifests as a progressive form by showing different rejection regions based on the values of auxiliary variable. Both methods are superior to existing methods without using 
any auxiliary variable. 
By taking into account the variability of the experiments and controlling for false discovery rate, the proposed methods can help to increase the confidence and reproducibility of gene expression analysis. 
We present numerical studies supporting our proposed methods and the real data example of 
Set4$\Delta$ mutant data. 
Our proposed methods can control a given level of FDR and are more powerful 
in detecting the true alternative hypotheses 
than one-stage procedures 
and some existing methods with covariate.    
Especially in the real data,  
two-stage FDR (S) 
is more efficient to select all the genes in the HXT gene family compared to all the other methods.

Overall, the proposed methods incorporating additional information 
into the problem of testing the primary variable are a valuable contribution to the field of gene expression analysis, and have the potential to improve the quality and reliability of experimental results.  
Additionally, the proposed methods can be widely applied beyond the gene-related data we introduced, to multiple testing problems 
where there exist both the primary and auxiliary variables.


\section{Data and Code Availability}
{The dataset(Set4$\Delta$) and source code developed for this research can be found at the following GitHub repository: \href{URL}{https://github.com/twostageFDR/.github}. 
This repository includes all scripts and instructions necessary to replicate the computational procedures and analyses presented in this paper.}

\section*{Conflict of Interest}
{The authors have declared no conflict of interest.}

\section*{Supporting Information}
Additional Supporting Information may be found in the online version of this article.

\section*{Acknowledgements}
This work was supported by the National Research Foundation of Korea (NRF) grant funded by the Korea government (MSIT) (No. 2020R1A2C1A01100526).

\bibliographystyle{abbrvnat} 
\bibliography{reference}

\label{lastpage}
\end{document}


\maketitle

\section{Technical Details}

\subsection{{Proof of Lemma 1}}
\begin{proof}
In order to prove the claim \eqref{eqn:claim}, we begin by observing that 
\begin{eqnarray}
P(P_i^H(\gamma_1) \le \gamma, P_{1i} \le \gamma_1) = P(P_i^H(\gamma_1) \le \gamma \given[] P_{1i} \le \gamma_1)\cdot P(P_{1i}\le \gamma_1).
\end{eqnarray}
We define $p_{i}^H(\gamma_1) = C(\gamma_1, p_{2i})$ when $p_{1i} \le \gamma_1$ in equation \eqref{eqn:gamma_def}.  
By employing the definition of $p_i$ and its equivalent formulation in terms of $\gamma$ in \eqref{eqn:gamma_def2}, we can establish the following equivalence:
\begin{eqnarray}
    \{P_i^H(\gamma_1) \le \gamma\} = \{C(\gamma_1, P_{2i}) \leq C(\gamma_1, \gamma_2)  \} = \{P_{2i} \le \gamma_2\}
\end{eqnarray}
since $C(\gamma_1,x)$ 
is increasing in $x$ for given $\gamma_1$. 
By substituting $\{P_i^H(\gamma_1) \le \gamma\}$ to $\{P_{2i} \le \gamma_2\}$, we obtain 
$P(P_i^H(\gamma_1) \le \gamma \given P_{1i} \le \gamma_1) = P(P_{2i} \le \gamma_2 \given P_{1i} \le \gamma_1)$  which  implies
$P(P_i^H(\gamma_1) \le \gamma, P_{1i} \le \gamma_1) = P(P_{2i} \le \gamma_2 , P_{1i} \le \gamma_1).$
\end{proof}

\subsection{Proof of Theorem 1}

\begin{proof}
Generally, it holds that
\begin{eqnarray}
P(P_i^H(\gamma_1) \le \gamma) = \underbrace{P(P_i^H(\gamma_1)\le \gamma, P_{1i} > \gamma_1)}_{\text{(a)}} + \underbrace{P(P_i^H(\gamma_1)\le \gamma, P_{1i} \le \gamma_1)}_{\text{(b)}}. 
\label{eqn:decomp}
\end{eqnarray}
If  $\gamma \leq \gamma_1$, 
we show that $(a) =0$ and $(b) =\gamma$. 
Since  $p_{1i} >\gamma_1$ implies  $P_i^H(\gamma_1)= p_{1i}$ from \eqref{eqn:p_def1},  
we have 
\begin{eqnarray*}
(a)=P(P_i^H(\gamma_1)\le \gamma, P_{1i} > \gamma_1) = P(P_{1i} \le \gamma, P_{1i} > \gamma_1 ) =0
\end{eqnarray*}
where the last equality is due to the fact that  the event  $\{P_{1i} \le \gamma, P_{1i} > \gamma_1 \}$  is empty due to   $\gamma \leq \gamma_1$. 
 By \eqref{eqn:claim} and the definition \eqref{eqn:gamma_def},  we have 
\begin{eqnarray*}
(b)=P(P_i^H(\gamma_1)\le \gamma, P_{1i} \le \gamma_1) = P(P_{2i} \le \gamma_2, P_{1i} \le \gamma_1) = \gamma.    
\end{eqnarray*}
Therefore, the sum of part (a) and part (b) is equal to $\gamma$ when $\gamma \leq \gamma_1$.\\
In the case of  $\gamma >\gamma_1$ we have   
\begin{eqnarray*} 
 (a)=P(P_i^H(\gamma_1)\le \gamma, P_{1i} > \gamma_1) = P(\gamma_1 \le P_{1i} \le \gamma) = \gamma - \gamma_1
 \end{eqnarray*}
 where we use the fact that $p_{i}^H=p_{1i}$ when $p_{1i}>\gamma_1$ and 
 $ \gamma>\gamma_1$. 
(b) Given $p_{1i} \le \gamma_1$, we can use the definition \eqref{eqn:gamma_def} above to obtain 
\begin{eqnarray}
    P_i^H(\gamma_1)= C(\gamma_1,p_{2i})
\end{eqnarray}
which implies that $P_i^H(\gamma_1)\leq C(\gamma_1, 1) = \gamma_1$.
Since $\gamma>\gamma_1$, the condition $P_i^H(\gamma_1)< \gamma$ always holds. Hence, 
we have 
\begin{eqnarray}
P(P_i^H(\gamma_1)< \gamma, P_{1i} \leq \gamma_1) = P(P_{1i} \leq \gamma_1) = \gamma_1.
\end{eqnarray}
From  $(a) =\gamma - \gamma_1$ and  $(b)=\gamma_1$, 
we can obtain the desired probability
$P(P_i^H(\gamma_1)\leq \gamma) = (a)+(b)  = \gamma.$ This completes the proof for $p_{i}$ of two-stage FDR(H).
\end{proof}

\subsection{Proof of Theorem 2}
\begin{proof}
We define  a function  $\alpha(u_1,x)$ satisfying  
\begin{eqnarray}
P(P_{2i} \le \alpha(u_1, x) | P_{1i} = u_1) = x
\label{eqn:alpha}
\end{eqnarray}
and denote $c(p_{2i}|p_{1i})$ as conditional probability density function of $p_{2i}$ given $p_{1i}$, 
 which is obtained from 
 $ c(\beta |y ) = 
 \frac{\partial  }{\partial \beta } C(\beta | y) = \frac{\partial }{\partial \beta}
  P(P_{2i}\leq \beta | P_{1i}=y)$.  
It is also obtained that 
$P(P_{2i} \leq z | P_{1i} =u_1 ) \leq x $ 
is equivalent to  $z \leq \alpha(u_1, x)$. 
Then  we have 
 $\{P_i^S \leq x  \}  =  
 \{ C(P_{2i}|P_{1i}) \leq x   \}  
 = \{P_{2i} \leq \alpha ( P_{1i},x) \}$. 
 Using $P_{1i} \sim U[0,1]$, i.e., $f_{P_{1i}}(p) = I_{[0,1]}(p)$, we have 
\begin{align*}
    P(P_{i}^S \le x) &= P(P_{2i} \leq \alpha (P_{1i},x)) \\
    &= \int_0^1 \underbrace{P(P_{2i} \leq \alpha(u_1, x) |P_{1i}=u_1 )}_{ \mbox{$=x$ by \eqref{eqn:alpha}}}   f_{P_{1i}} (u_1)  du_1   \\
                &= \int_{0}^{1} x du_1\\
                &= x
\end{align*}
for $x\in [0,1]$. 
Therefore,  $p_i^S$ of two-stage FDR(S) also follows uniform distribution in $[0,1]$. 
\end{proof}

\section{Copula Choice}
\label{sec:copula}
\subsection{Examples of Copula}
\label{sec:copula_example}
Copula allows us to accurately model complex relationships between variables. For example, consider the scatterplots shown in Figure \ref{fig:copula_example}, where all variables have the same marginal distribution, $U(0,1)$. Despite this uniform marginal distribution, the different patterns observed in the scatter plots result from variations in their joint distributions. 
We choose an appropriate copula using model selection criteria. 

\begin{figure}[ht]
  \centering
  \includegraphics[width=\textwidth]{supple/figure_s1.pdf}
    \caption{Random vectors with different copulas\citep{joe1997multivariate}. All these random vectors share the same marginal distribution ($U(0,1)$) and exhibit the same correlation of $-0.4$, except for the independent copula. The subfigures represent different copulas: (A) Independent, showing no correlation; (B) Normal, illustrating standard Gaussian dependency; (C) Frank; (D) Clayton (rotated 90°); (E) Gumbel (rotated 270°); (F) Joe (rotated 270°).}
    \label{fig:copula_example}
\end{figure}

\subsection{Comparing Gaussian, Frank, Clayton, Gumbel and Joe copulas with simulation}
\label{sec:copula_sim}
Table \ref{tbl:copula_test} shows the outcome of selecting a copula model 100 times under the assumption that the true underlying model is the Clayton copula. The findings highlight that models other than the true Clayton copula can be incorrectly chosen based on criteria such as Log-Likelihood (LogLik), Akaike Information Criterion (AIC), and Bayesian Information Criterion (BIC).

\begin{table}[ht]
\centering
\caption{The number of selected copula based on LogLik, AIC and BIC when generating random variables using the Clayton copula with mean and standard deviation.}
\label{tbl:copula_test}
\resizebox{\textwidth}{!}{%
\begin{tabular}{cccccccccc}
\hline
         & \multicolumn{3}{c}{LogLik}     & \multicolumn{3}{c}{AIC}          & \multicolumn{3}{c}{BIC}          \\ \cline{2-10} 
Family   & \# selected & mean     & std.   & \# selected & mean      & std.    & \# selected & mean      & std.    \\ \hline
Gaussian & 0          & 1615.157 & 57.301 & 0          & -3228.315 & 114.602 & 0          & -3221.328 & 114.602 \\
Frank    & 0          & 1524.409 & 53.864 & 0          & -3046.818 & 107.729 & 0          & -3039.830 & 107.729 \\
Clayton  & 100        & 2194.153 & 71.345 & 100        & -4386.307 & 142.690 & 100        & -4379.320 & 142.690 \\
Gumbel   & 0          & 1090.186 & 50.598 & 0          & -2178.372 & 101.197 & 0          & -2171.385 & 101.197 \\
Joe      & 0          & 501.012  & 35.078 & 0          & -1000.024 & 70.156  & 0          & -993.037  & 70.156  \\ \hline
\end{tabular}%
}
\end{table}

\section{More Figures of Real Data Analysis}


\subsection{Obtained $p_i^H$ and $p_i^S$ of Real Data}
Based on the marginal $p$-values $p_{1i}$ and $p_{2i}$, 
  we  compute $p_i^H$ and $p_i^S$. 
  Their histograms are presented in in Figure \ref{fig:hist_pval} showing typical cases of 
  $p$-values 
  for sparse cases.

\begin{figure}[ht]
    \centering
    \includegraphics[width =\textwidth]{supple/hist_pval.pdf}
    \caption{Histograms of (A) $p_i^H(0.987)$ and (B) $p_i^S$.}
    \label{fig:hist_pval}
\end{figure}

\subsection{Venn Diagram of the numbers of significant genes}
Figure \ref{fig:venn} summarizes the numbers of significant genes which consists of the numbers rejected by two different one-stage procedures and two proposed two-stage procedures. Notably, most of genes that are significant at two-stage FDR(H) are also significant at two-stage FDR(S), although the results are not entirely consistent. 

\begin{figure}[ht]
    \centering
    \includegraphics[width=0.8\linewidth]{supple/venn.png}
    \caption{The plot represents a Venn diagram comparing the rejected hypotheses from one-stage FDR(locfdr), one-stage FDR(Storey), two-stage FDR(H) and two-stage FDR(S).}
    \label{fig:venn}
\end{figure}

\section{Comparison with DESeq2}
\label{sec:DESeq2}

In the previous section, our method utilizes a copula model to analyze the relationship between logfold changes and their standard deviations in RNA sequencing data. This differentiates our approach from traditional methods like EdgeR in \cite{robinson2010edger} and DESeq2 in \cite{love2014moderated}, which adjust logfold changes and their standard deviations separately without directly considering their relationship.

DESeq2 in \cite{love2014moderated}, for instance, processes RNA sequencing data in a systematic manner. It normalizes the data to make it comparable across genes and samples, estimates dispersion through a technique called shrinkage, and conducts hypothesis testing to identify genes with significant expression differences under various conditions, such as knockout (KO) versus wild type (WT). A critical part of DESeq2’s methodology involves recalculating the logfold changes and their standard deviations based on updated dispersion parameters.

For clarity, we explore how DESeq2 calculates these dispersion parameters. Initially, DESeq2 normalizes the raw data, setting the stage for further analysis. It then models the gene expression counts using a negative binomial (NB) distribution:

\begin{align}
    KO_{ij} \text{ or } WT_{ij} &\sim NB(\mu_{ij}, a_{i}), \\
    \log(\mu_{ij}) &= 
    \begin{cases}
        x_j + \beta_i &\text{if }KO\\
        x_j &\text{if }WT
    \end{cases}
    \label{eqn:deseq2}
\end{align}
where $\mu_{ij}$ represents the expected mean expression of gene $i$ in sample $j$, and $a_{i}$ is the dispersion parameter for gene $i$. The $\beta_i$ replaces the logfold change for gene $i$. Unlike directly estimating $a_i$ by maximum likelihood estimation, DESeq2 models the logarithm of the dispersion parameter, $\log(a_{i})$, as being normally distributed:

\begin{equation}
    \log(a_{i}) \sim N\left(\frac{b_1}{\bar{\mu}_{i}} + b_2, \sigma_d \right) \label{eqn:ai}
\end{equation}
where $\bar{\mu}_{i}$ is the average of normalized counts for gene $i$, and $b_1$, $b_2$ are parameters estimated via maximum likelihood estimation. Subsequently, DESeq2 re-estimates $a_i$ 
using the maximum a posteriori (MAP) estimate of $a_i$. Utilizing this re-estimated $a_i$, DESeq2 fits the model in \eqref{eqn:deseq2} and conducts tests on $\beta_i$.

To compare our approach directly with DESeq2, we calculated our auxiliary variable, the MAP of $a_i$, as in DESeq2.
Then we employed a Gaussian copula with a parameter of $\tau = -0.4$. This approach led to the rejection of 504 genes using the two-stage(H) and 588 genes using the two-stage(S) at a significance level of $\alpha = 0.10$. Notably, among the 37 genes rejected by DESeq2, 28 genes were also identified by our two-stage(S), illustrating a significant overlap and highlighting the effectiveness of incorporating the standard deviation of logfold changes into our analysis. 
Compared to our proposed methods, DESeq2 suggested by \citet{love2014moderated} has two distinctive features: Firstly, they do not solely depend on raw logfold change values to reject hypotheses. Our method directly estimates the standard deviation of logfold changes by bootstrap, and DESeq2 estimates with multi-step approach. Secondly, our approach uses copulas to model the relationship between logfold changes and their standard deviation, contrasting with DESeq2's implementation of Negative Binomial regression. Despite DESeq2 not explicitly outlining a rejection region based on dispersion, Figure \ref{fig:ds2}(A) shows how gene rejection varies with the re-estimated dispersion parameter, indicating an implicit dependency.

The Loglik, AIC and BIC for the candidate copulas are summarized in Table \ref{sup_tbl:copula_choice}, and the rejections of DESeq2, Two-stage (H), and Two-stage (S) with the same dispersion parameter with DESeq2 are illustrated in Figure \ref{fig:ds2}.

\begin{figure}[ht]
    \centering
    \includegraphics[width=0.85\linewidth]{supple/figure_s4.pdf}
    \caption{Comparative Analysis of Gene Rejection Methods. Left column: log-scaled standard deviation vs. logfold changes. Right column: auxiliary p-value (from DESeq2) vs. primary p-value (from the main text). Panels: (A) DESeq2, (B) Two-stage (H), (C) Two-stage (S). Red dots indicate rejected genes for each method.}
    \label{fig:ds2}
\end{figure}

\begin{table}[ht]
\centering
\caption{The Log-likelihood, AIC, BIC of Gaussian, Clayton, Gumbel and Joe copulas with set4$\Delta$ dataset with estimated standard deviations of logfold changes as in DESeq2}
\begin{tabular}{@{}cccc@{}}
\toprule
Family   & LogLik          & AIC               & BIC               \\ \midrule
Gaussian & \textbf{1578.79}         & \textbf{-3155.57}          & \textbf{-3148.56}          \\
Frank    & 1486.99         & -2971.97          & -2964.96          \\
Clayton  & 836.13 & -1670.25 & -1663.24 \\
Gumbel   & 1199.29         & -2396.58          & -2389.56          \\
Joe      & 486.19          & -970.38           & -963.37           \\ \bottomrule
\end{tabular}
\end{table}

\bibliographystyle{plainnat}  
\bibliography{reference}